\documentclass[aps,pra,twocolumn,superscriptaddress]{revtex4-2}
\usepackage{amsfonts,bbm,graphicx,times}
\usepackage[pdftex]{color}
\usepackage[utf8]{inputenc}
\usepackage{graphicx}
\usepackage[colorlinks, linkcolor=blue, citecolor=blue, urlcolor=blue, breaklinks=true]{hyperref}
\usepackage{braket}
\usepackage{soul}
\usepackage{autobreak}
\usepackage{algpseudocode}
\usepackage{algorithm}
\usepackage{microtype}

\usepackage{amsthm}
\usepackage{csquotes}

\usepackage{amsthm}\newcommand{\be} {\begin{equation}}
\usepackage{bbold}
\usepackage{dsfont}
\newcommand{\ee} {\end{equation}}

\newcommand{\Tr}{{\rm Tr}}

\newtheorem{theorem}{Theorem}

\usepackage{color}
\usepackage{mathrsfs}
\usepackage{amssymb}
\newcommand{\bla}{bla\\bla\\bla\\bla\\bla}

\DeclareUnicodeCharacter{2283}{\supset}
\usepackage{ragged2e}

\begin{document}
\title{Optimal measurement-based quantum thermal machines in a finite-size system}

\author{Chinonso Onah}
\affiliation{Department of Physics, RWTH Aachen University, Germany}
\affiliation{Volkswagen AG, Berliner Ring 2, Wolfsburg 38440, Germany}
\author{Obinna Uzoh}
\affiliation{Department of Physics, Simon Fraser University, 8888 University Drive, Burnaby, B.C. Canada}
\author{Obinna Abah}
\affiliation{School of Mathematics, Statistics and Physics, Newcastle University, Newcastle upon Tyne NE1 7RU, United Kingdom}

\begin{abstract}
We present a measurement-based quantum thermal machine that extracts work from the back-action of generalized quantum measurements whose working medium is a coupled two-level quantum system. Specifically, we derive universal optimization criteria for a three-stroke measurement-based engine cycle with coupled two-level system of Ising-like interaction as a working medium. Furthermore, we present two numerical algorithms to optimize the engine work extraction and enhance its performance. Our numerical results demonstrate:
(i) efficiency peaks in the projective-measurement limit;
(ii) symmetry breaking (detuning or weak coupling) enlarges the  exploitable energy gap; and
(iii) performance remains robust ($>50\%$ of optimum) under $\sim\!10^\circ$ feedback-pulse errors. The framework is platform-agnostic and directly implementable with current superconducting, trapped-ion, or NMR technologies, providing a concrete route to scalable, measurement-powered quantum thermal machines.
\end{abstract}

\maketitle

\section{Introduction}
Thermal machines have been at the heart of thermodynamics since the early days of the field \cite{Callen1985,Deffner2019book}. They are often used to provide mechanical energy from thermal energy. At microscopic scales, energy exchange processes are profoundly influenced by quantum coherence, correlations, and measurement backaction, which requires a reexamination of fundamental thermodynamic concepts such as work, heat, and machine performance \cite{Myers}. Recent advances in nanofabrication techniques have led to the realization of coherent quantum thermal devices \cite{Koski2015PRL,Martinez2016,Rossnagel2016,Josefsson2018,Lindenfels2019PRL}. Moreover, understanding the thermodynamics of quantum systems is an important topic at the intersection of quantum information, nonequilibrium physics, and emerging quantum technologies \cite{RevModPhys.93.035008}. 

Another quantum feature gaining prominence in nonequilibrium thermodynamics is quantum measurement, which, beyond its traditional function as a means of information extraction, can act as an active dynamical process that injects energy, induces dissipation, and modifies system-environment interactions. Recent studies have shown that quantum measurements can play a constructive role in the operation of quantum thermal machines \cite{Yi2017,Elouard2017,Myers}. This has led to growing interest in measurement-based quantum thermal machines, where measurements -- information and feedback -- are exploited as thermodynamic resources to enable refrigeration, work extraction, or heat pumping \cite{Elouard2017,Yi2017,Perna2024}. These kinds of thermal devices have been studied with different working substances, such as coupled-qubits~\cite{Bhandari2022}, two spins \cite{Das2018}, and two spin-1/2 interacting via the anisotropic XXZ Heisenberg interaction \cite{Anka2021PRE,Behzadi2025}.   

The interplay between quantum measurement and Maxwell's demon arrow of time has been examined through the framework of the quantum measurement-based engine \cite{yanik22}. The coefficient of performance for the modified refrigerator through generalized measurement has been shown to increase linearly with the measurement strength parameter, beyond the classical cooling of the familiar Otto cycle refrigerator \cite{Behzadi2024}.
The realization has been explored in the proof-of-concept experiment using a nuclear magnetic resonance setup \cite{Lisboa2022PRA}, and later became a basis for investigating the interplay between coherent control and an indefinite causal order \cite{Dieguez2023}.

Moreover, it has been demonstrated that suppressing quantum coherences may enhance the performance of measurement-based engines \cite{Lin2021}. It is shown that there is the possibility of circumventing the degradation effects induced by coherence, which enhance the work extraction and efficiency of the engine, by having sufficient control over the measurement basis angles.  That is, an appropriate choice of the measurement basis can effectively act as a form of \emph{quantum lubrication} that suppresses the detrimental coherences generated during the engine cycle operation. More generally, controlling both measurement bases and applied unitary operations can significantly enhance the performance of measurement-driven thermal machines \cite{Wang2024PLA}.
Beyond coherence control, additional quantum resources have also been identified as mechanisms for performance enhancement. For instance, considering the performance of measurement engines with different initial states, it is demonstrated that nonideal measurement engines operating with thermally correlated configurations outperform their entanglement-based counterparts \cite{Abhisek}. In a recent study, correlations between spin pairs have been demonstrated to enable modulation of efficiency through the choice of measurement basis, highlighting the important role of quantum correlations in optimizing quantum thermal machines \cite{Rathnakaran2025}. Furthermore, the performance of measurement-driven engines depends explicitly on the strength parameters of the measurement process, which are determined by the coupling constants between the working substance and the measurement apparatus \cite{Behzadi2025}. 

Motivated by these developments, in this work we investigate the optimization criteria for a measurement powered engine cycle whose working medium is a finite-size quantum system coupled by Ising-like interaction.  The engine cycle comprises three-strokes of measurement, feedback, and thermalization for a system initially prepared in the thermal state. By constructing the feedback Hamiltonian, we determine the criteria for optimal feedback angles for maximum work extraction and efficiency. We present two numerical algorithms, hybrid local-global search and simple grid search, to obtain the optimal feedback angles. We further examine the impact of on-site symmetry on the performance of the measurement-based engine.

The rest of the paper is organized as follows: In Sec. \ref{sec:protocol} we briefly present the operation of the finite size coupled two-level quantum systems cycle using weak measurements. We further present the work extraction and performance of the cycle when it functions as an engine. In Sec. \ref{Sec-feedback-hamiltonian}  we construct the feedback Hamiltonian. The optimal feedback angles corresponding with the extraction of maximum work are derived in Sec. \ref{sec:problem}. In Sec. \ref{sec:examples} we present the exact solvable examples of single qubit and two-body coupled quantum systems.  We present the numerical methods in Sec. \ref{sec:algorithms_angles}, with conclusions in Sec. \ref{sec:conclusion}.

\begin{figure}[!t]
\begin{center}
\includegraphics[width=0.98\columnwidth]{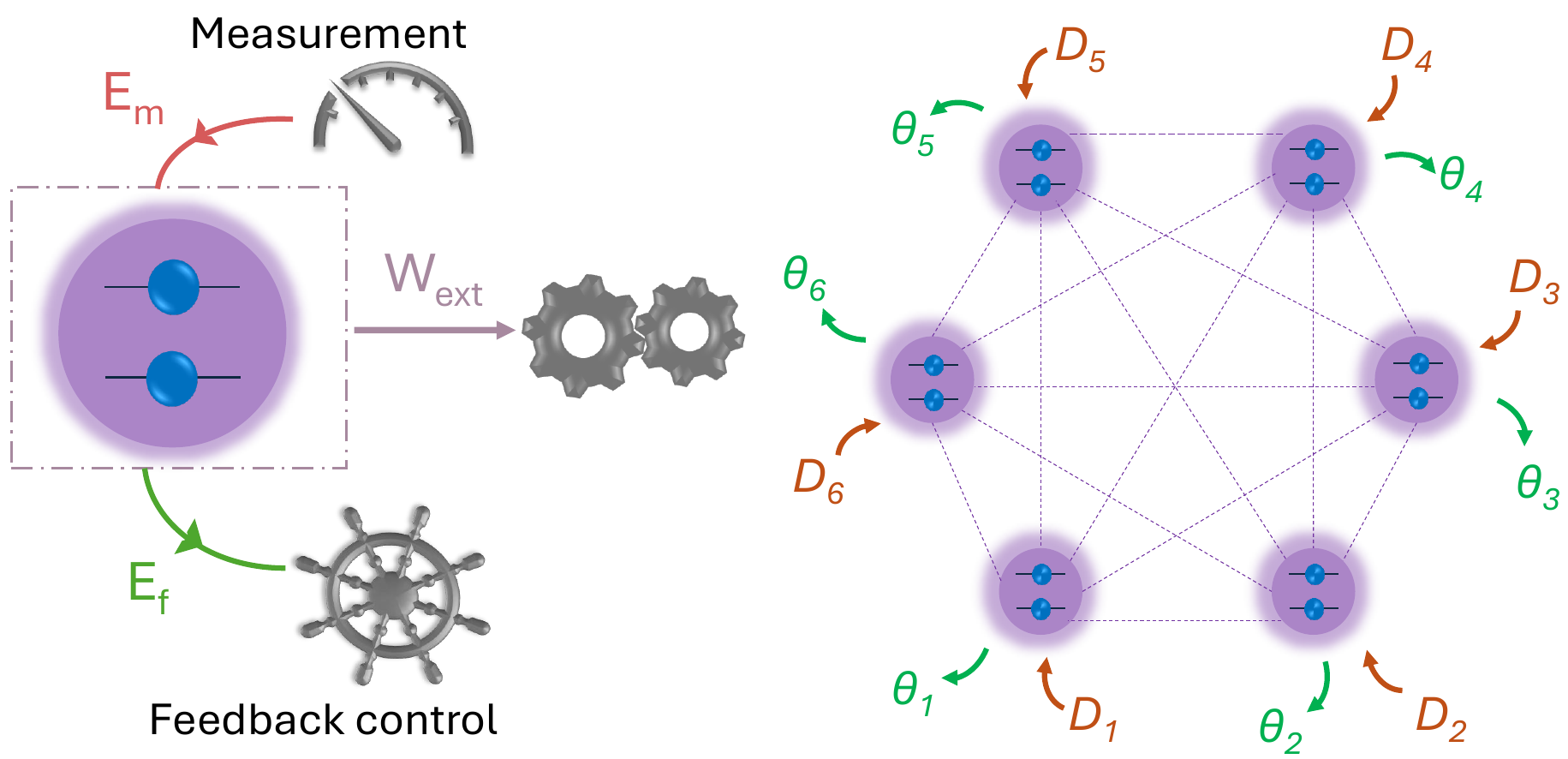}
\caption{\textit{Left panel:} Schematic illustration of feedback driven measurement-based quantum machine with a two-level system.  \textit{Right panel:} Thermal device working substance composing of a six coupled two-level system. Each two-level system (skyblue circle) is initially aligned along the $z$-axis and the gray dashed lines represent couplings between the two-level systems. Red arrows indicate weak $\sigma_x$ measurements by detectors $D_i$, and green arcs labeled $\theta_i$ represent local feedback (rotations around $y$-axis) applied based on measurement outcomes.}
\label{fig1}
\end{center}
\end{figure}



\section{Measurement-based quantum engine }\label{sec:protocol}
We consider a quantum thermal machine based on quantum measurement whose working substance is a generic finite number of coupled two-level quantum systems, see Fig. \ref{fig1}. 
The system is described by Hamiltonian 
\begin{align}
H_S\!=\!\frac{1}{2}\,\mathbf{I}_{\otimes N}
\,+\,
\sum_{j=1}^{N}\frac{\epsilon_j}{2}\,\sigma_z^{j}
+ \sum_{j<k}^{N}\Delta^{jk}_{zz}\,\sigma_z^{j}\,\sigma_z^{k},
\label{eq:HQ_n_qubits_shifted}
\end{align}
where the indices $j$ identify single two-level quantum systems that make up the quantum thermal machine. The stages of the cycle can be described as follows.

\textit{Initialization} -- First, the system is initially prepared at thermal equilibrium with a thermal reservoir at temperature $T$. The initial state reads $\rho_{th}\!=\!\exp{-\beta H_S}/Z$, where $Z\!=\!\mbox{tr}(e^{-\beta H_S})$ is the partition function and the corresponding average energy is given by $E_i\!=\!\mbox{tr}(\rho_{th} H_S)$. 

\textit{Measurement} -- In the second stroke, the system is isolated from the thermal reservoir and then interacts with the probe or measurement device $D$ that performs a general measurement. The measurement is described by a positive operator-valued measure (POVM) with the corresponding Kraus operators {$M_k$}. The resulting post-measurement state conditioned on outcome \(k\) for the system reads
\begin{equation}
\rho_{M_k} = \frac{M_k \rho_{th} M_k^\dagger}{p_k},
\label{eq:post}
\end{equation}
where $p_k\!=\!\mbox{Tr} [M_k \rho_{th} M_k^\dagger ]$.  For instance, when discrete weak \(\sigma_x\) measurements are performed on the quantum system $H_S$ using the probe(s) \(D_i\) (for \(i = 1, 2, \ldots\)), they are characterized by the Kraus operators \(\{M_{i+}, M_{i-}\}\), given by
\begin{align}
    M_{i\pm} &= \tfrac12 \bigl(\sqrt{\kappa_i} + \sqrt{1-\kappa_i}\bigr) \,\mathbf{I}_{2\otimes N}\nonumber\\
    &\pm\,\tfrac12\bigl(\sqrt{\kappa_i} - \sqrt{1-\kappa_i}\bigr)
    \bigl(\mathbf{I}_{2\otimes(i-1)}\otimes\sigma_x\otimes \mathbf{I}_{2\otimes(N-i)}\bigr).
\end{align}
The measurement operators satisfy the completeness relation, 
\begin{equation}
M_{i+}^{\dagger}M_{i+}+M_{i-}^{\dagger}M_{i-}=\mathbf{I}_{2\otimes N},\qquad i=1,\dots,N.
\end{equation}
where $\mathbf{I}_{2\otimes N}$ is the identity matrix of $N$ coupled two-level system. In the analyses to follow, all quantities are evaluated branch by branch at fixed \(\kappa\). For this reason, we suppress the explicit \(\kappa\)-dependence and branch label in the notation and use \(\rho_M\) to denote the post measurement state. Thus, the quantum system energy after measurement is $E_m\!=\!\mathrm{Tr}(H_S \rho_M).$ 

\textit{Feedback control} -- The third process, which follows after measurement, is a local application of a unitary feedback operation $U_j(\theta_j)$, which is parameterized by the angle $\theta_j$. At the end of the process, the post-feedback state is denoted by the density matrix as
\begin{equation}
\rho_{F}
\;=\;
U(\{\theta_j\})
\,\rho_{M}\,
U^\dagger(\{\theta_j\}),
\end{equation}
with the corresponding energy of the system after feedback given by $E_F = \mbox{Tr}(\rho_F H_S)$. 

\textit{Thermalization} -- Finally, put the system in contact with the heat reservoir to allow heat  exchange until it returns to its initial thermal state $\rho_{th}$.
To maintain cyclic operation, the engine requires resetting/erasing the final state $\rho_F$ back to the initial thermal equilibrium state $\rho_{th}$. The minimal thermodynamic cost associated with this erasure process is bounded by Landauer’s principle \cite{RevModPhys.93.035008}. Specifically, the erasure work $W_{\mathrm{er}}$ is
\begin{equation}
W_{\mathrm{er}} = k_B T\, D(\rho_F \| \rho_{th}),
\label{eq:work_erasure}
\end{equation}
where $D(\rho\|\sigma)\!=\! \mathrm{Tr}\left[\rho(\ln\rho - \ln\sigma)\right]$ denotes the quantum relative entropy.

\textit{Performance} -- Let us now consider the work exchanges and efficiency for this measurement-based engine. The performance of the thermal machine is embodied in the total output work extraction per a given cycle defined by
\begin{equation}
W_{ext} = E_m - E_F= \mathrm{Tr}\bigl[\,H_S\bigl(\rho_{M} - \rho_{F}\bigr)\bigr].
\end{equation}
Given these definitions, the efficiency $\eta$ of the quantum measurement-based engine is defined as the ratio between the net useful work output (after accounting for erasure costs) and the energy provided at the measurement stage $E_m$:
\begin{equation}
\eta = \frac{W_{\mathrm{ext}} - W_{\mathrm{er}}}{E_m}.
\label{eq:efficiency_general}
\end{equation}
Therefore, the efficiency takes the explicit form:
\begin{equation}
\eta = \frac{\mathrm{Tr}\left[H_S(\rho_M - \rho_F)\right] - k_B T D(\rho_F\|\rho_{th})}{\mathrm{Tr}(H_S \rho_M)}.
\label{eq:efficiency_explicit}
\end{equation}




\section{Constructing the feedback Hamiltonian}\label{Sec-feedback-hamiltonian}
To simplify the analysis of the optimal feedback problem, let  $ H_{S,N}$ represent $N$ fully connected quantum measurement engines for which we are to determine the optimal feedback angles for maximum work extraction after a measurement process. Given that the post-measurement state is known, we observe that the optimal feedback angles can be determined by applying arbitrary unitary transformation on the system Hamiltonian to simulate the inverse of the feedback operations and solving the resulting optimization problem to determine the angles.  To see this, one exploits the cyclicity of trace to express the energy after feedback in terms of the post measurement state, the system Hamiltonian and arbitrary angles as follows  
\begin{align}
    E_F=\mbox{tr}(\rho_F H_S) \equiv \mbox{tr}(\rho_M H_F(\{\theta_j\})),
\end{align}
where the feedback Hamiltonian $H_F$ is given by
\begin{equation}\label{eq:HF}
H_F(\{\theta_j\}) = \bigotimes_{j=1}^{N} U_j^\dagger(\theta_j)
\left( \sum_{j=1}^{N} H_j + \sum_{j<k} H_{jk} \right) \bigotimes_{j=1}^{N} U_j(\theta_j).
\end{equation}
Here, each term in the sums represents a single-body Hamiltonian $H_j$ or a two-body interaction Hamiltonian $H_{jk}$, resulting from the corresponding feedback unitaries \(U_j(\theta_j)\). The feedback unitary can be local or global, and the operators satisfy the constraint $U_j^\dagger(\theta_j) U_j(\theta_j) = \mathbf{I}, \quad \forall \; j=1,2,\dots, N.$


Considering the rotation around the $y$ axis, 
the local feedback unitary rotation \(U_j(\theta_j)\) generated by $\sigma_y$ reads
\begin{align}
U_j^\dagger(\theta_j) \sigma_x^{j} U_j(\theta_j) &= \sigma_x^{j} \cos\theta_j + \sigma_z^{j} \sin\theta_j,  \, \mathrm{and} \nonumber\\
U_j^\dagger(\theta_j) \sigma_z^{j}U_j(\theta_j) &= -\sigma_x^{j} \sin\theta_j + \sigma_z^{j} \cos\theta_j, \label{eq:U_transformations2}
\end{align}
where $\sigma_i$ ($i=x,y,z$) are the Pauli matrices.
On the other hand, the global (non-local) feedback operator is given by 
\begin{equation}
U_G(\theta) = \exp\left(-i \theta\, \sigma_y^{1} \sigma_y^{2}\right) = \cos \theta\, \mathbf{I} - i \sin \theta\, \sigma_y^{1} \sigma_y^{2}.
\end{equation}
We remark that the local feedback protocol outperforms the global feedback protocol; see the Appendix \ref{global_appendix}. 


\section{Optimization problem: Optimal feedback}\label{sec:problem}
Our goal is to determine the set of angles \(\{\theta_j\}\) that minimizes the energy functional \(E_F\). Formally, the optimization problem is stated as:
\begin{equation}
\min_{\{\theta_j\}} E_F(\{\theta_j\}) =\mbox{tr} \Bigl[ \rho_M H_F(U(\theta_j))\Bigr],
\end{equation}
subject to $U_j^\dagger(\theta_j) U_j(\theta_j) = I,\, \forall j = 1,2,\dots,N.$

The solution to this problem is given by the set of optimal feedback angles \(\{\theta_j^*\}\)  corresponding to the stationary points of the energy surface, which  satisfy
\begin{equation}
\frac{\partial E_F}{\partial \theta_j} = 0,\quad \forall \; j=1,2,\dots,N\,.
\label{eq:partial_theta_j_general}
\end{equation}
By inspecting the collection of second derivatives that forms the Hessian matrix \(\mathcal{H}\) with elements
\begin{equation}
\mathcal{H}_{jk} = \frac{\partial^2 E_F}{\partial \theta_j\, \partial \theta_k},
\end{equation}
we can determine if the stationary point is a minimum, maximum, or saddle point. Thus, it is often straightforward to evaluate the corresponding energy values of the stationary angles to determine the optimal feedback angles \(\{\theta_j^\ast\}\).

We now address the central result of this paper, that is, to obtain the optimal feedback angles that result in the extraction of maximum work from a measurement-based quantum thermal machine of a finite-size system. Assuming that the working medium of the quantum machine is a $N$-coupled two-level system, the central result is presented in the theorem and followed by some examples.

\begin{theorem}[Stationarity Conditions for the Optimal Feedback Angles]
\label{thm:optimal_feedback_induction}

Consider a quantum measurement-based machine of finite size consisting of $N$ coupled two-level systems described by the Ising-like Hamiltonian of the form \cite{Padurariu2010,Pita-vidal2024}:
\begin{align}
H_S\!=\!\frac{1}{2}\,\mathbf{I}_{\otimes N}
\,+\,
\sum_{j=1}^{N}\frac{\epsilon_j}{2}\,\sigma_z^{j}
+ \sum_{j<k}^{N}\Delta^{jk}_{zz}\,\sigma_z^{j}\,\sigma_z^{k},
\end{align}
where $\mathbf{I}_{\otimes N}$ is the identity matrix of the $N$-space, $\epsilon_j$ is the energy of qubit $j$, $\sigma_z^{j}$ is the Pauli-$Z$ operator on qubit $j$, and $\Delta_{zz}$ denotes the coupling between qubits $j$ and $k$.

After a projective or POVM-based measurement on the quantum system, the post-measurement state is $\rho_{M}$.  Then, apply local feedback unitaries,
\begin{align}
U(\{\theta_j\}) \;=\; \bigotimes_{j=1}^{N} U_j(\theta_j),
\end{align}
with
$U_j(\theta_j)\!=\!\exp\!\Bigl(-i\,\theta_j\,\sigma_y^{j}/2\Bigr)$.
The post-feedback state is
\begin{align}
\rho_{F}
\;=\;
U(\{\theta_j\})
\,\rho_{M}\,
U^\dagger(\{\theta_j\}).
\end{align}
The optimal local feedback angles $\{\theta_j^{\ast}\}_{j=1}^N$ satisfy a set of self-consistent tangent equations of the form
\begin{align}
\tan \theta_j
\;=\;
\frac{-B_j(\theta_1,\ldots,\theta_{j-1},\theta_{j+1},\ldots,\theta_{N})}%
     {A_j(\theta_1,\ldots,\theta_{j-1},\theta_{j+1},\ldots,\theta_{N})},
\label{eq:theta_j_general_form_shifted}
\end{align}

\begin{equation}
\theta_j^\ast
=
\operatorname{atan2}\!\bigl(-B_j,A_j\bigr)
\qquad (\mathrm{mod}\,\pi).
\label{eq:atan2_general}
\end{equation}

for $1 \,\le j \,\le N$, with:
\begin{equation}
A_j(\{\theta_{k\neq j}\})
\;=\;
\frac{\epsilon_j}{2}\,\Bigl\langle \widetilde{\sigma_{z}^{j}}\Bigr\rangle
\;+\;\sum_{k\neq j}\Delta_{zz}^{jk}\;\Bigl\langle\widetilde{\sigma_{z}^{j}\sigma_{z}^{k}}\Bigr\rangle,
\end{equation}
\begin{equation}
B_j(\{\theta_{k\neq j}\})
\;=\;
\frac{\epsilon_j}{2}\,\Bigl\langle \widetilde{\sigma_{x}^{j}}\Bigr\rangle
\;+\;\sum_{k\neq j}\Delta_{zz}^{jk}\;\Bigl\langle\widetilde{\sigma_{x}^{j}\sigma_{z}^{k}}\Bigr\rangle.
\end{equation}

The “ \(\widetilde{\phantom{\sigma}}\)” indicates partial rotations on the other qubits (see Appendix~\ref{thm:optimal_feedbak}, for proof). $\Box$
\end{theorem}

The theorem determines the stationary-angle equations that generate the candidate feedback angles. The optimal feedback is then obtained by evaluating the feedback energy over these stationary solutions and selecting the minimizing branch through Algorithm~\ref{alg:extended_optimal_feedback} or Algorithm~\ref{alg:gridsearch_feedback}.

\begin{figure*}[!ht]
\includegraphics[width=1\textwidth]{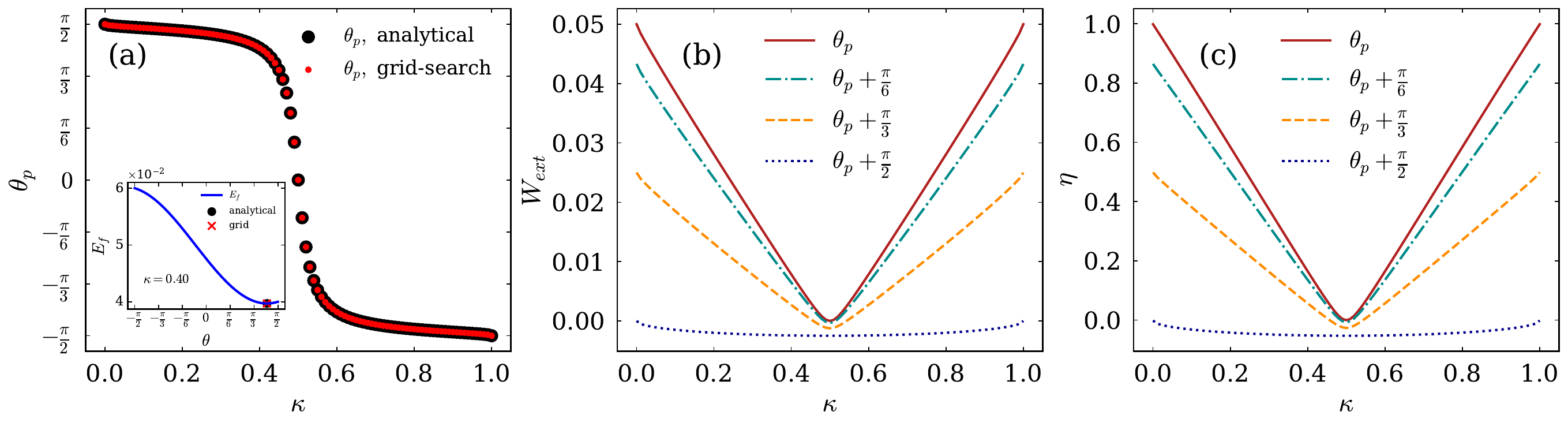}
\caption{Left to right: Optimal feedback angle, maximum work extraction, and efficiency as a function of measurement strength parameter for a single qubit measurement-based engine. The inset is the energy functional as a function of optimal angle with the red-cross corresponding to the optimal angle deduced from the grid search algorithm.}
\label{fig2}
\end{figure*}
\section{Examples}\label{sec:examples}
Here, we present explicit derivations of the optimal angle(s) for two cases, that is, a single qubit, $N=1$ and two coupled qubits  measurement-based engines, $N=2$.

\subsection{Single qubit measurement-based engine}
\label{subsec:single_measurement}

Consider the measurement-based quantum engine whose working medium is described by the Hamiltonian \footnote{An elementary algebraic treatment of the single engine case based on transformation properties of the Pauli matrices is included in the appendix~\ref{sec:single_engine_analysis}.};
\begin{align}
H_S =  \frac{1}{2} \mathbf{I_2} + \frac{\epsilon_1}{2} \sigma_z^{1}.
\end{align}
The feedback protocol is of the form $U\!=\!\exp(-i\,\theta_1 \sigma_y^{1}/2),$, so
\begin{equation}
H_F = \frac{1}{2}\,  \mathbf{I_2} + \frac{\epsilon}{2} (\sigma_z \cos \theta - \sigma_x \sin \theta) .
\end{equation}
Defining the expectation values, $\langle \sigma_z \rangle\!=\!a$ and
$\langle \sigma_x \rangle\!=\!b$, the energy of the quantum system after feedback operation is
\begin{equation}
E_F =  \frac{1}{2} + \frac{\epsilon}{2} \, a \cos \theta - \frac{\epsilon}{2} \, b \sin \theta.
\label{eq:EF}
\end{equation}

The optimal feedback angle $\theta^\ast$ associated with the maximum extraction of the work output corresponds to the minimum feedback energy $\min_{\theta} (E_F)$. Minimizing $\frac{dE_F}{d\theta}\!=\!0$, the  optimal feedback angle \( \theta^* \) reads \cite{Kurt2003}
\begin{equation}
\theta^\ast = \arctan\left(-b/a\right).\footnote{\noindent
For the general theorem we adopt the branch-sensitive form
\begin{equation}
\theta_j^\ast
=
\operatorname{atan2}(-B_j,A_j)
\qquad (\mathrm{mod}\,\pi),
\end{equation}
whereas in the explicit one- and two-qubit examples we keep the simpler relations
\(
\tan\theta^\ast=-\frac{b}{a},
\qquad
\tan\theta_i^\ast=-\frac{B_i}{A_i},
\)
which are obtained directly from the elementary stationary equations. The physically relevant solution is then identified by selecting the branch that minimizes \(E_F\).}
\label{eq:theta_optimal}
\end{equation}
Figure \ref{fig2} shows the feedback angle, work extraction and efficiency as a function of measurement strength. Figure \ref{fig2}(a) shows the correspondence between the analytical and grid search solution of the feedback angles. The inset is the feedback energy as a function of the optimal feedback angle when the measurement strength is $\kappa=0.40$. The total extraction of work per cycle and the efficiency for different phase shifts in the optimal feedback angle are presented in Figs. \ref{fig2}(b) and (c), respectively. We observe that increasing the angle phase shift decreases the work extraction and the engine efficiency.


\subsection{ Two-body measurement-based engine}
Considering the case of measurement-based quantum engine using a two-coupled quantum system as its working substance. Specifically, the system Hamiltonian is in the form,
\begin{equation}
H_S =  \frac{1}{2}\mathbf{I}_{2\otimes 2} + \frac{\epsilon_1}{2} \sigma_z^{1} + \frac{\epsilon_2}{2} \sigma_z^{2}  + \Delta_{z}^{12}\, \sigma_z^{1} \sigma_z^{2}, 
\label{eq:2QME}
\end{equation}
where $\sigma_z^j$ represents the Pauli operator acting on qubit $j\!=\!1,2$; the coefficients $\epsilon_j$ are the energy of qubit $j$ and $\Delta_z^{12}\!=\!\Delta_z$ represents the coupling energy between qubits $1$ and $2$. Recently, a single superconducting spin  qubit \cite{Yoshihara2021} as well as supercurrent-mediated coupling between two distant spin qubits have been experimentally explored \cite{Pita-Vidal2023} and with all-to-all connected spin qubits in Ref. \cite{Pita-vidal2025}.
Assuming $\epsilon_1\!=\!\epsilon_2\!=\!0.5$ in Eq.~\eqref{eq:2QME}, the energy eigenvalues read,
\begin{equation}
\begin{aligned}
E_{00}&=1+\Delta_z,\\
E_{01}&=E_{10}=0.5-\Delta_z,\\
E_{11}&=\Delta_z.
\end{aligned}
\label{eq:eigenvals}
\end{equation}
There exist pairwise degeneracies that occur at specific values of $\Delta_z$; namely $\Delta_z\!=\!+0.25$ where $E_{01}\!=\!E_{11}$, and $\Delta_z\!=\!-0.25$ where $E_{01}\!=\!E_{00}$. Let the difference between the first excited and ground energies be defined as the gap $\Delta E$, which can be  distinguished into three regions as follows;
\begin{equation}
\Delta E(\Delta_z)=
\begin{cases}
1,               & \Delta_z\le -0.25,\\
0.5-2\Delta_z,   & -0.25\le \Delta_z\le +0.25,\\
2\Delta_z-0.5,   & \Delta_z\ge +0.25.
\end{cases}
\label{eq:gap_two_spin}
\end{equation}

\begin{figure}[!]
\includegraphics[width=0.94\columnwidth]{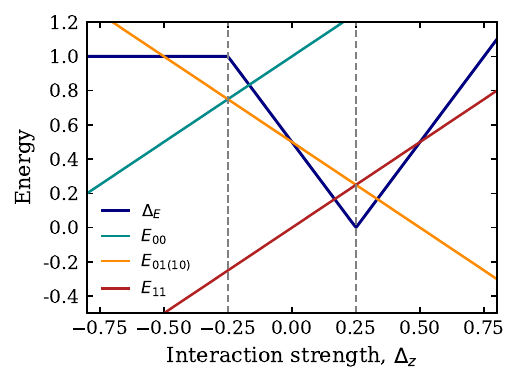}
 \caption{Energy spectrum and gap of the two-coupled quantum spin system with $\epsilon_1\!=\!\epsilon_2\!=\!0.5$ as a function of the interaction strength $\Delta_z$.
 Energy gap $\Delta E$ as a function of $\Delta_z$. For $\Delta_z\le-0.25$ (ferromagnetic branch) the gap saturates at $\Delta E=1$; inside the window $-0.25\le\Delta_z\le0.25$ it varies
linearly; for $\Delta_z\ge0.25$ (antiferromagnetic branch) it grows unbounded as $\Delta E\simeq 2\Delta_z$.}
\label{fig3}
\end{figure}
Figure~\ref{fig3} shows the energy eigenvalues and the gap \(\Delta E\) as a
function of \(\Delta_z\). The gap is continuous at \(\Delta_z=\pm 0.25\), with
a flat branch only for \(\Delta_z\le -0.25\). In contrast, for
\(\Delta_z\ge 0.25\) the gap grows linearly with slope \(2\). Thus, for
\(0<\Delta_z<0.25\) the gap still decreases with increasing \(\Delta_z\), and
only beyond the crossing at \(\Delta_z=0.25\) does the antiferromagnetic branch
become the branch with increasing gap. These kinks reflect the level crossings
of the longitudinal Ising spectrum in the absence of transverse terms \cite{Sachdev2011, Majer2005, PetterssonFors2024}. In this symmetric case, the degeneracy in the Hamiltonian results in stationary points that correspond to the maximum energy difference, $E_m-E_F(\theta)$.

In the limit of zero coupling, the optimal feedback operation of the quantum measurement-based machine will correspond to the discussion in Ref. \cite{Bhandari2022}. 
Based on Theorem~\ref{thm:optimal_feedback_induction}, the optimal feedback angles satisfy;
\begin{equation}
\theta_i = \arctan \left( \frac{-B_i}{A_i} \right), \quad i=1,2;\footnote{\noindent
For the general theorem we adopt the branch-sensitive form and the physically relevant solution is then identified by selecting the branch that minimizes \(E_F\).} 
\label{eq:optimal_angles}
\end{equation}
where 
\begin{eqnarray}
A_1 &=& \frac{\epsilon_1}{2}\,a_1 + \Delta_z^{12}\,\left(c_{zz}\,\cos\theta_2 - c_{zx}\,\sin\theta_2\right),\\
B_1 &=& \frac{\epsilon_1}{2}\,b_1 +
\Delta_z^{12}\,\left[c_{xz}\,\cos\theta_2 - c_{xx}\,\sin\theta_2\right],\\
A_2 &=& \frac{\epsilon_2}{2}\,a_2 +
\Delta_z^{12}\,\left[c_{zz}\,\cos\theta_1 - c_{xz}\,\sin\theta_1\right],\\
B_2&=& \frac{\epsilon_2}{2}\,b_2 + \Delta_z^{12}\,\left[c_{zx}\,\cos\theta_1 - c_{xx}\,\sin\theta_1\right],
\end{eqnarray}
with $\langle\sigma_z^{i}\rangle\!=\!a_i$,  $\langle\sigma_z^{1}\sigma_z^{2}\rangle\!=\!c_{zz}$, $\langle\sigma_z^{1}\sigma_x^{2}\rangle\!=\!c_{zx}$, $\langle\sigma_x^{i}\rangle\!=\!b_i$, $\langle \sigma_x^{1}\sigma_z^{2}\rangle\!=\!c_{xz}$, and $\langle\sigma_x^{1}\sigma_x^{2}\rangle\!=\!c_{xx}.$

Following the algorithm protocol to obtain the optimal feedback angle in Section \ref{sec:algorithms_angles}, we present the feedback energy as a function of angles based on two different algorithms in Fig. \ref{fig:energy_surface}. 
Then, using the resulting optimal energy after feedback, we can evaluate the work output and efficiency of two-qubit quantum measurement-based engines.
\begin{figure}[!h]
\includegraphics[width=0.94\columnwidth]{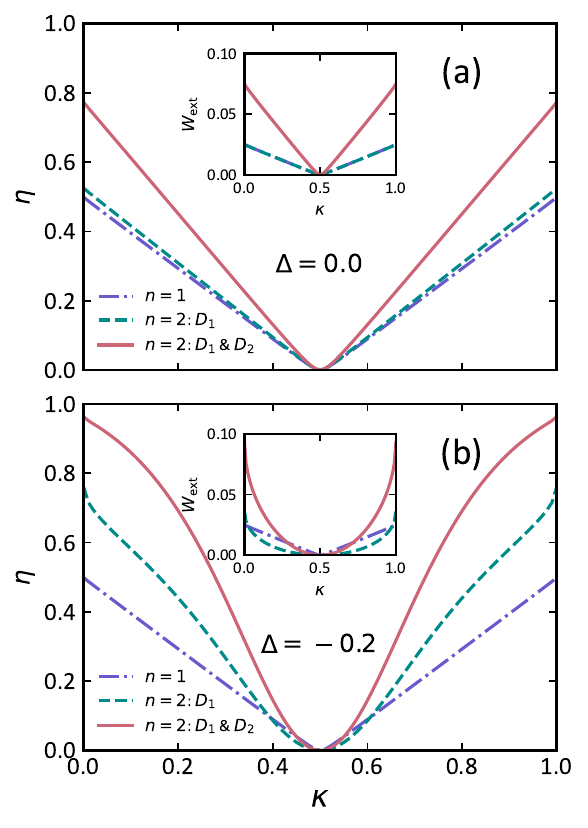}
 \caption{Efficiency as a function of measurement strength for two qubits measurement-based engine at optimal feedback angles for three engine configuration, (a) $\Delta\!=\!0.0$ and (b) $\Delta\!=\!-0.2$. A single-qubit engine ($n=1$) is represented by purple dotted-dashed curve, a two-qubit engine with a single detector ($n=2:D_1$)  is teal dashed curve, and a two-qubit engine with double detectors ($n=2:D_1$ \text{and} $D_2$) is denoted with pink solid curve. \emph{Insets}: Output work as function of measurement strength. Parameters are $\epsilon_1\!=\!-0.05$ and $\epsilon_2\!=\!-0.10$.}
\label{figure4}
\end{figure}
Figure~\ref{figure4} shows the efficiency $\eta$ and the extracted work output $W_{\text{ext}}$ (shown in the insets) of the quantum measurement-based engines as a function of the measurement strength ($\kappa$). The analysis compares three distinct engine configurations; a single-qubit engine ($n=1$), a two-qubit engine with a single detector ($n=2:D_1$) and a two-qubit engine with two detectors ($n=2:D_1$ \text{and} $D_2$).

Figure~\ref{figure4}~(a) illustrates the dynamics of the engine when there is no coupling between the qubit systems, $\Delta\!=\!0.0$. All configurations exhibit strict symmetry around an intermediate measurement strength of $\kappa = 0.5$, at which point both efficiency and work output vanish completely. This indicates that no useful work can be extracted in this exact measurement regime. The performance of the measurement-based engine peaks at the weak ($\kappa \to 0$) and strong ($\kappa \to 1$) measurement limits. 
The two-qubit engine utilizing two detectors demonstrates a clear thermodynamic advantage, reaching a peak efficiency of approximately 0.8, which significantly outperforms the other configurations. Furthermore, the inset reveals a degeneracy in the work output: at $\Delta = 0.0$, the $n=1$ and $n=2: D_1$ configurations overlap perfectly, yielding identical extracted work output despite their different theoretical setups. Figure~\ref{figure4}~(b) depicts the systems response for non-zero qubit coupling, $\Delta\!=\!-0.2$. Although the fundamental symmetric structure and the zero-performance point at $\kappa\!=\!0.5$ are preserved, the shift in $\Delta$ remarkably improves the maximum capabilities of the engine. In particular, the two-qubit engine with two detectors ($n\!=\!2:D_1$ and $D_2$) achieves  optimal efficiency $\eta \approx 0.95$ at measurement extremes. In addition, non-zero $\Delta$ breaks the degeneracy observed in Figure~\ref{figure4}~(a). The one-site energy was set to $\epsilon_1=-0.05$ and $\epsilon_2=-0.10$, to achieve the improved engine performance of utilizing two measurement detectors over one measurement detector.

Moreover, breaking the on-site symmetry, $\epsilon_1\neq\epsilon_2$, clearly increases the energy gap in which the feedback operation can be exploited to enhance the extractable work of the measurement-based engine . That is, the feedback energy associated with the symmetry-breaking scenario, $E_F^\ast(\xi)$, is less than the symmetric case, where $\xi\!=\!\varepsilon_2-\varepsilon_1$. Therefore, extractable work and efficiency are enhanced by non-trivial symmetry breaking, i.e., $W_{\max}(\xi\!\neq\!0)>W_{\max}(\xi\!=\!0)$ and $\eta_{\max}(\xi\!\neq\! 0)\!>\!\eta_{\max}(\xi\!=\!0)$.


\begin{figure}[!th]
    \centering
    \includegraphics[width=0.95\columnwidth]{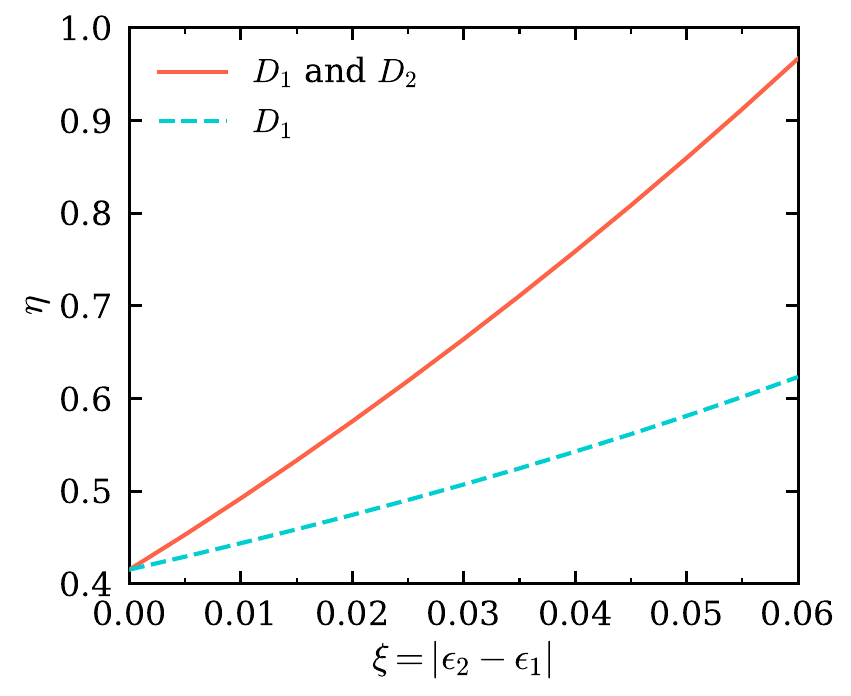}
    \caption{Efficiency $\eta$ as a function of  detuning $\xi$  
    for the two coupled qubits measurement-based engine at fixed qubit coupling $\Delta\!=\!-0.20$ and measurement strength $\kappa\!=\!0.10$.   The red curve shows the case of one detector  while the green curve correspond to the case of using two detectors ($D_1$ and $D_2$).}
    \label{fig:sym_eff}
\end{figure}

Figure~\ref{fig:sym_eff} presents numerical results for the efficiency of the measurement-based engine as a function of the symmetry breaking parameter associated with the on-site energy of the working substance. We observe that engine efficiency increases with increasing difference in the energy gap of the composite coupled finite quantum system.
Similar monotonic behavior has been observed in measurement-powered engines using superconducting qubits \cite{Bresque2021}, aligning with the broader principle that breaking symmetry increases the available non-equilibrium free energy \cite{Adlam2022,Myers2020PRE}. This reflects findings from two-qubit engines \cite{Bresque2021} and entropic demons \cite{Elouard2017}, reinforcing the idea that symmetry acts as a thermodynamic resource; its breaking enables feedback mechanisms to extract additional free energy.




\section{Algorithms for the angles optimisation}
\label{sec:algorithms_angles}

In this section, we present numerical procedures for estimating feedback angles
that minimize the post-feedback energy \(E_F(\boldsymbol\theta)\). The resulting
optimization landscape is generally nonconvex. An exhaustive grid search on an
\(M^N\) lattice scales as \(\mathcal{O}(M^N)\), which rapidly becomes
impractical as \(N\) grows. For this reason, we use two complementary
procedures: a self-consistent fixed-point refinement initialized from coarse
grid seeds, and a pure grid-sweep baseline for small systems.
However, for small- to moderate-sized systems, high-quality and often certified optimal angles are obtained using a two-tier optimization strategy that combines local gradient-based descent with a global grid-based exploration.  

\paragraph{Hybrid local–global search
(Alg.~\ref{alg:extended_optimal_feedback}).}
The routine starts from a coarse Cartesian grid
$\theta_j\in[-\pi,\pi]$ with spacing
$\Delta\theta\simeq10^{-1}\text{–}10^{-2}$.
Each grid point is polished by one step of the
self-consistent update
$\operatorname{atan2}\!\bigl(-B_j,A_j\bigr)
\qquad (\mathrm{mod}\,\pi)$
Eq.~\eqref{eq:theta_j_general_form_shifted}. The left–arrow “$\leftarrow$’’ in the algorithms is \emph{computer-science} notation,
  meaning “overwrite the current value of $\theta_j$ with the right-hand
  expression’’ and  $\vec\theta_{\!\backslash j}$ denotes the vector of all angles
  \emph{except} the $j$-th one, i.e.\
  $\vec\theta_{\!\backslash j}=(\theta_1,\ldots,\theta_{j-1},
  \theta_{j+1},\ldots,\theta_N)$.
 This update is a self-consistent fixed-point / block-coordinate refinement based
on the stationary-angle equations in
Eq.~\eqref{eq:atan2_general}. In our numerical experiments it converges
rapidly for the single- and two-qubit examples when initialized from a coarse
grid. We then cluster nearby stationary points (using an \(\ell_\infty\)-metric
with tolerance \(10^{-3}\)) and rank them by the resulting feedback energy.
For the two-qubit case, the best cluster coincides with the minimum found by the
brute-force grid sweep to numerical precision.

\begin{figure}[!h]
\includegraphics[width=1\columnwidth]{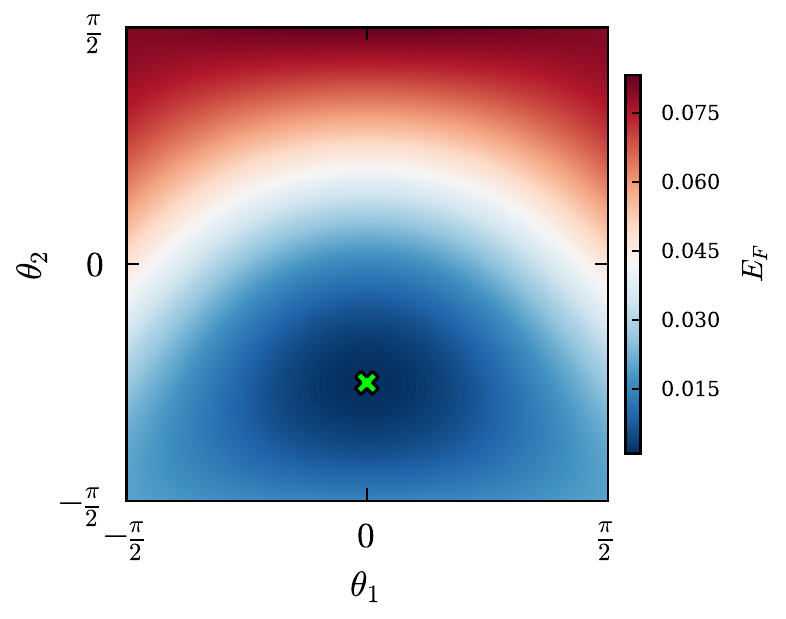}
 \caption{Feedback energy landscape, $E_F(\theta_1, \theta_2)$ for two coupled-qubit measurement-based engine. 
The \emph{green cross} marks the solution found by a local, iterative algorithm that updates \(\theta_1\) and \(\theta_2\) based on self-consistent arctan formulas, which coincides with the global minimum identified from a brute-force grid search over \((\theta_1, \theta_2)\).
The systems parameters are $\epsilon_1\!=\!0.05$, $\epsilon_2\!=\!0.10$, $\Delta\!=\!-0.2$, and $\kappa\!=\!0.2$. For this case of $\kappa$, the work extracted is~$2.07\times10^{-2}$ with an efficiency of~$0.69$.}
\label{fig:energy_surface}
\end{figure}


\paragraph{Pure grid-sweep baseline
(Alg.~\ref{alg:gridsearch_feedback}).}
For small systems we also use a uniform grid sweep over
\((\theta_1,\dots,\theta_N)\), evaluate \(E_F\) (or equivalently the gradient
norm \(\|\nabla E_F\|\)) on that grid, and refine candidate stationary points by
a local quadratic fit. Although this brute-force procedure scales as
\(\mathcal{O}(M^N)\) and is therefore limited to very small \(N\), it provides a
useful numerical baseline for the single- and two-qubit examples studied here.



In fully connected arrays of spins, the number of local minima increases exponentially, -- a behavior characteristic of the spin-glass system \cite{Aarts1989} -- rendering the traditional grid-based search intractable. Nevertheless, a carefully designed combination of local descent, global enumeration, and symmetry reduction enables near-optimal feedback angles for all instances studied, while maintaining scalability to the medium-sized systems relevant to current quantum thermodynamic platforms.

\begin{algorithm}[H]
\caption{Hybrid seed-and-refine search for stationary feedback angles}
\label{alg:extended_optimal_feedback}
\begin{algorithmic}[1]
\Require System parameters; correlators entering \(A_j,B_j\); coarse grid
\(\theta_j\in[\theta_{\min},\theta_{\max}]\); tolerance \(\varepsilon\); maximum
iterations \(K_{\max}\).
\State \(\mathcal{S}\gets \varnothing\)

\For{each coarse-grid seed \(\boldsymbol\theta^{(0)}\)}
    \State \(\widetilde{\boldsymbol\theta}\gets \boldsymbol\theta^{(0)}\)
    \For{\(k=1\) to \(K_{\max}\)}
        \For{\(j=1\) to \(N\)}
            \State
            \(
            \widetilde{\theta}_j
            \gets
            \operatorname{atan2}\!\bigl(
            -B_j(\widetilde{\boldsymbol\theta}_{\backslash j}),
            A_j(\widetilde{\boldsymbol\theta}_{\backslash j})
            \bigr)
            \)
        \EndFor
        \If{\(\|\widetilde{\boldsymbol\theta}^{(k)}-\widetilde{\boldsymbol\theta}^{(k-1)}\|_\infty<\varepsilon\)}
            \State \textbf{break}
        \EndIf
    \EndFor
    \State \(\mathcal{S}\gets \mathcal{S}\cup
    \{(\widetilde{\boldsymbol\theta},E_F(\widetilde{\boldsymbol\theta}))\}\)
\EndFor

\State Cluster nearby angle vectors in \(\mathcal{S}\)
\State Return the clustered stationary points and the one with minimal \(E_F\)
\end{algorithmic}
\end{algorithm}

\begin{algorithm}[H]
\caption{Simple Grid Search}
\label{alg:gridsearch_feedback}
\begin{algorithmic}[1]
\Require System parameters \(\epsilon_1\), \(\epsilon_2\), \(\Delta\); post-measurement state \(\rho_M\); angle range \(\theta_{\mathrm{min}}, \theta_{\mathrm{max}}\); grid size \(N\); gradient tolerance \(\epsilon_{\mathrm{tol}}\).
\Statex

\State \(\Delta\theta \gets (\theta_{\mathrm{max}}-\theta_{\mathrm{min}})/(N-1)\)
\State \(\mathcal{C} \gets \varnothing\) \Comment{Initialize candidate list for stationary points}

\For{$i=0$ \textbf{to} $N-1$}
    \For{$j=0$ \textbf{to} $N-1$}
        \State \(\theta_1 \gets \theta_{\mathrm{min}} + i\,\Delta\theta\)
        \State \(\theta_2 \gets \theta_{\mathrm{min}} + j\,\Delta\theta\)

        \Comment{Compute the partial derivatives of the feedback energy}
        \State \(G_{\theta_1} \gets \dfrac{\partial E_F}{\partial \theta_1}(\theta_1,\theta_2)\)
        \State \(G_{\theta_2} \gets \dfrac{\partial E_F}{\partial \theta_2}(\theta_1,\theta_2)\)
        \State \(\|\nabla E_F\| \gets \sqrt{G_{\theta_1}^{2}+G_{\theta_2}^{2}}\)

        \Comment{Mark as candidate if gradient norm is below threshold}
        \If{$\|\nabla E_F\| < \epsilon_{\mathrm{tol}}$}
            \State \(\mathcal{C} \gets \mathcal{C} \cup \{(\theta_1,\theta_2)\}\)
        \EndIf
    \EndFor
\EndFor

\Comment{Group nearby candidates into clusters of stationary points}
\State \(\mathcal{S} \gets \mathrm{Cluster}(\mathcal{C})\)

\Comment{Select a representative from each cluster}
\State \(\mathcal{P} \gets \varnothing\)
\For{$C \in \mathcal{S}$}
    \State \((\theta_1^\ast,\theta_2^\ast) \gets \mathrm{Representative}(C)\)
    \State \(\mathcal{P} \gets \mathcal{P} \cup \{(\theta_1^\ast,\theta_2^\ast)\}\)
\EndFor

\Comment{Evaluate feedback energy at each stationary-point representative}
\State \(\mathcal{E} \gets \varnothing\)
\For{$(\theta_1^\ast,\theta_2^\ast) \in \mathcal{P}$}
    \State \(E^\ast \gets E_F(\theta_1^\ast,\theta_2^\ast)\)
    \State \(\mathcal{E} \gets \mathcal{E} \cup \{(\theta_1^\ast,\theta_2^\ast,E^\ast)\}\)
\EndFor

\Comment{Identify the global minimum from the stationary-point set}
\State \((\theta_1^{\mathrm{opt}},\theta_2^{\mathrm{opt}},E^{\mathrm{opt}})\gets
\operatorname*{arg\,min}_{(\theta_1^\ast,\theta_2^\ast,E^\ast)\in\mathcal{E}} E^\ast\)

\State \Return \(\mathcal{E}\) and \((\theta_1^{\mathrm{opt}},\theta_2^{\mathrm{opt}},E^{\mathrm{opt}})\)
\end{algorithmic}
\end{algorithm}


\section{Conclusions}\label{sec:conclusion}

We have developed a framework for optimal measurement-based feedback in finite quantum measurement engines composed of coupled two-level systems. While a single qubit quantum measurement engine can be returned to its original Bloch orientation to maximize work extraction, we show that coupling fundamentally alters the optimization objective. For coupled finite-size systems, the optimal feedback strategy is more naturally defined by maximizing the net energy transfer rather than restoring the local state orientation.
To determine the optimal feedback parameters, we derived explicit conditions for the feedback angles and implemented complementary numerical approaches.
An iterative scheme was used to solve the coupled feedback equations, while a grid-search method ensured that all stationary solutions were identified. Our results show that the interplay between the coupling strength and the measurement parameters generates a nontrivial energy landscape for the feedback angles, with multiple local extrema of the feedback energy.


The proposed framework is directly relevant to experimental platforms where measurement and feedback are accessible, including superconducting qubits in circuit-QED architectures~\cite{Cottet2017}, trapped-ion systems, and quantum-dots arrays with tunable couplings. Moreover, the approach is naturally extendable to continuous-variable platforms, such as cavity optomechanical modes, where measurement-induced backaction and feedback control can be engineered with high precision. Finally, our findings provide a systematic route for optimizing feedback protocols in measurement-driven quantum thermal machines and open new possibilities for exploiting correlations and measurement backaction to enhance thermodynamic performance in both discrete and continuous quantum systems.

\appendix
\section*{Appendix}

\section{Optimal feedback of a measurement-based engine}
\label{sec:single_engine_analysis}
Here, we present the pedagogical analysis of a quantum measurement-based engine whose working substance is a two-level system. Starting with the initial density matrix of the system $\rho_i\!=\!\tfrac{1}{2}\bigl(\mathbf{I} + \vec\sigma\cdot\hat r_i\bigr)$ and taking the post-measurement density matrix $\rho_M=\tfrac{1}{2}\bigl(\mathbf{I} + \vec\sigma\cdot\hat r_m\bigr)$, where $\vec\sigma$ is the Pauli matrix and $r_j (j=i,m)$ is the direction vector. 
Then, after implementing an arbitrary feedback protocol with a rotation by angle~$\theta$ about the unit axis~$\hat n$, the resulting density matrix is as follows,
\begin{equation}
  \rho_f 
  = e^{-i\frac{\theta}{2}\,\hat n\cdot\vec\sigma}\,
    \rho_M\,
    e^{\,i\frac{\theta}{2}\,\hat n\cdot\vec\sigma}.
  \label{eq:feedback}
\end{equation}
Using the relation $(\vec r\cdot\vec\sigma)\,(\hat n\cdot\vec\sigma) = \vec r\cdot\hat n + i\,\vec\sigma\cdot(\vec r\times\hat n)$ and the rotation identity, Eq.~(\ref{eq:feedback}) becomes
\begin{equation}
  \rho_f
  = \tfrac{1}{2}\,\bigl(\mathbf{I} + \vec{r}_f \cdot \vec{\sigma}\bigr)
  \label{eq:rho_f_final}
\end{equation}
where
\begin{align}
 \vec r_f\cdot\vec\sigma
  &= e^{-i\frac{\theta}{2}\,\hat n\cdot\vec\sigma}\,
    (\hat r_m\cdot\vec\sigma)\,
    e^{\,i\frac{\theta}{2}\,\hat n\cdot\vec\sigma}\nonumber\\
  &=
  \Bigl[
    \hat n(\hat n\cdot\hat r_m)
    + \cos\theta\bigl(\hat r_m-\hat n(\hat n\cdot\hat r_m)\bigr) \nonumber\\
   & + \sin\theta\,(\hat n\times\hat r_m)
  \Bigr]\!\cdot\!\vec{\sigma}.
  \label{eq:rf_def}
\end{align}



Taking the co-ordinates of the Bloch sphere after measurement $\hat r_m=x_+\hat x+z_+\hat z$, and assuming that the feedback rotation is around the $y$ axis (i.e. $\hat n=\hat y$) with a global phase $\phi=0$,  the general form of the rotated Bloch vector is
\begin{align}
  \vec r_f
  &= \cos\theta\,\vec r_m
    + \sin\theta\,(\hat y\times\vec r_m)\nonumber\\
  &= (x_m\cos\theta+z_m\sin\theta)\,\hat x
    +(z_m\cos\theta-x_m\sin\theta)\,\hat z\,.
  \label{eq:bloch_length}
\end{align}
The resulting feedback density matrix is
\begin{equation}
\rho_f
  = \tfrac12\bigl(\mathbf{I} + \sigma_x\,(x_+\cos\theta+z_+\sin\theta)
  +\sigma_z\,(z_+\cos\theta - x_+\sin\theta)\bigr).
\end{equation}


The optimal feedback condition requires that the protocol outcome $\rho_f$ aligns with the $z$ axis up to magnitude $|z_f|$, i.e. (see, \cite{Lloyd2000,Kurt2003,yanik22})
\begin{equation}
  \rho_f
  = \tfrac12\bigl(I - |z_f|\,\sigma_z\bigr)
  = \tfrac12\bigl(I + \vec r_f\cdot\vec\sigma\bigr),
  \label{eq:optimal_condition}
\end{equation}
\begin{equation}
  -|z_f|\,\sigma_z
  = (x_+\cos\theta+z_+\sin\theta) \sigma_x\,
  + (z_+\cos\theta - x_+\sin\theta) \sigma_z.
  \label{eq:yz_reduction}
\end{equation}
From the $\sigma_x$ coefficient   we get $x_+\cos\theta+z_+\sin\theta=0$, then 
\begin{equation} \tan\theta = -\frac{x_+}{z_+}, 
\end{equation} 
$x_+$ and $z_+$ correspond, respectively, to $b$ and $a$ in Eq. (\ref{eq:theta_optimal}).

\section{Local transformation of single-body Pauli matrices}\label{sec:local_paulis}
Under the unitary rotation \(U_j(\theta_j)\) generated by \(\sigma_y\), the Pauli matrices transform as, using
\begin{equation}
U_j(\theta_j)=\cos\!\left(\frac{\theta_j}{2}\right)I
-i\sin\!\left(\frac{\theta_j}{2}\right)\sigma_y^j,
\end{equation}
\begin{equation}
U_j^\dagger(\theta_j)=\cos\!\left(\frac{\theta_j}{2}\right)I
+i\sin\!\left(\frac{\theta_j}{2}\right)\sigma_y^j,
\end{equation}
together with the Pauli commutation relations, one obtains
\begin{equation}
\begin{aligned}
U_j^\dagger(\theta_j)\,\sigma_x^{j}\,U_j(\theta_j)
&=\sigma_x^{j}\cos\theta_j+\sigma_z^{j}\sin\theta_j,\\
U_j^\dagger(\theta_j)\,\sigma_z^{j}\,U_j(\theta_j)
&=-\,\sigma_x^{j}\sin\theta_j+\sigma_z^{j}\cos\theta_j.
\end{aligned}
\end{equation}
This gives the identities used in the construction of the feedback Hamiltonian.

\subsection{Local transformation of two-body Pauli matrices}
Since all feedback unitaries are local, the following transformations are calculated for the two-body terms:


(a). \textbf{Transforming} \(\sigma_z^{1} \otimes \sigma_z^{2}\):
\begin{equation}
\begin{aligned}
U_2(\theta_2)^\dagger U_1(\theta_1)^\dagger \left( \sigma_z^{1} \otimes \sigma_z^{2} \right) U_1(\theta_1) U_2(\theta_2)
&= \\
\sigma_x^{1} \sigma_x^{2} \sin\theta_1 \sin\theta_2 - \sigma_x^{1} \sigma_z^{2} \sin\theta_1 \cos\theta_2 \\
 - \sigma_z^{1} \sigma_x^{2} \cos\theta_1 \sin\theta_2  + \sigma_z^{1} \sigma_z^{2} \cos\theta_1 \cos\theta_2.
\end{aligned}
\end{equation}

(b). \textbf{Transforming} \(\sigma_z^{1} \otimes \sigma_x^{2}\) 
\begin{equation}
\begin{aligned}
U_2^\dagger(\theta_2) U_1^\dagger(\theta_1) \left(\sigma_z^{1} \otimes \sigma_x^{2} \right) U_1(\theta_1) U_2(\theta_2) =\\
-\sigma_x^{1} \sigma_x^{2} \sin\theta_1 \cos\theta_2 
 - \sigma_x^{1} \sigma_z^{2} \sin\theta_1 \sin\theta_2 \\
  + \sigma_z^{1} \sigma_x^{2} \cos\theta_1 \cos\theta_2 
 + \sigma_z^{1} \sigma_z^{2} \cos\theta_1 \sin\theta_2.
\end{aligned}
\end{equation}

(c).  \textbf{Transforming} \(\sigma_x^{1} \otimes \sigma_z^{2}\):
\begin{equation}
\begin{aligned}
U_2^\dagger(\theta_2) U_1^\dagger(\theta_1) \left(\sigma_x^{1} \otimes \sigma_z^{2} \right) U_1(\theta_1) U_2(\theta_2) =\\
-\sigma_x^{1} \sigma_x^{2} \cos\theta_1 \sin\theta_2  + \sigma_x^{1} \sigma_z^{2} \cos\theta_1 \cos\theta_2 \\
 - \sigma_z^{1} \sigma_x^{2} \sin\theta_1 \sin\theta_2 
 + \sigma_z^{1} \sigma_z^{2} \sin\theta_1 \cos\theta_2.
\end{aligned}
\end{equation}

(d). \textbf{Transforming} \(\sigma_x^{1} \otimes \sigma_x^{2}\):
\begin{equation}
\begin{aligned}
U_2^\dagger(\theta_2) U_1^\dagger(\theta_1) \left(\sigma_x^{1} \otimes \sigma_x^{2} \right) U_1(\theta_1) U_2(\theta_2)=\\
\sigma_x^{1} \sigma_x^{2} \cos\theta_1 \cos\theta_2 + \sigma_x^{1} \sigma_z^{2} \cos\theta_1 \sin\theta_2 \\
+ \sigma_z^{1} \sigma_x^{2} \sin\theta_1 \cos\theta_2 + \sigma_z^{1} \sigma_z^{2} \sin\theta_1 \sin\theta_2.
\end{aligned}
\end{equation}

\subsection{Global transformation of Pauli matrices} \label{sec:global}
Let $U_G(\theta)$ be defined as:
\begin{equation}
U_G(\theta) = \exp\left(-i \theta\, \sigma_y^{1} \sigma_y^{2}\right) = \cos \theta\, I - i \sin \theta\, \sigma_y^{1} \sigma_y^{2}.
\label{global_unitary}
\end{equation}
Then the following transformations hold:
\begin{align}
\sigma_z^{1}(\theta) &= U_G^\dagger(\theta)\, \sigma_z^{1}\, U_G(\theta) = \cos 2\theta\, \sigma_z^{1} - \sin 2\theta\, \sigma_x^{1} \sigma_y^{2},\\ 
\sigma_z^{2}(\theta) &= U_G^\dagger(\theta)\, \sigma_z^{2}\, U_G(\theta)\!=\!\cos 2\theta\, \sigma_z^{2} - \sin 2\theta\, \sigma_y^{1} \sigma_x^{2}.\\
 \sigma_z^{1}(\theta)\, \sigma_z^{2}(\theta) &= U_G^\dagger(\theta)\, \sigma_z^{1} \sigma_z^{2}\,U_G(\theta)
    = \sigma_z^{1} \sigma_z^{2}.
\end{align}

\begin{proof}


    

Let
\begin{equation}
A:=\sigma_y^{1}\sigma_y^{2},
\qquad
U_G(\theta)=e^{-i\theta A}.
\end{equation}
Since \(A^2=\mathbf I\), we may write
\begin{align}
U_G(\theta)=\cos\theta\,\mathbf I-i\sin\theta\,A \nonumber \\  U_G^\dagger(\theta)=\cos\theta\,\mathbf I+i\sin\theta\,A.
\label{eq:UG_expand}
\end{align}

\paragraph{Transformation of \(\sigma_z^{1}\).}
Set \(B=\sigma_z^{1}\). Using \(\sigma_y\sigma_z=i\sigma_x\) and
\(\sigma_z\sigma_y=-i\sigma_x\), we obtain
\begin{align}
AB
&=\sigma_y^{1}\sigma_y^{2}\sigma_z^{1}
=(\sigma_y^{1}\sigma_z^{1})\sigma_y^{2}
=i\,\sigma_x^{1}\sigma_y^{2},\\
BA
&=\sigma_z^{1}\sigma_y^{1}\sigma_y^{2}
=-i\,\sigma_x^{1}\sigma_y^{2}.
\end{align}
Hence \(A\) and \(B\) anticommute,
\begin{equation}
\{A,B\}=AB+BA=0,
\end{equation}
and therefore
\begin{equation}
ABA=-B.
\label{eq:ABA_minusB}
\end{equation}

Now substitute Eq.~\eqref{eq:UG_expand} into
\(U_G^\dagger(\theta)\,B\,U_G(\theta)\):
\begin{align}
U_G^\dagger(\theta)\,B\,U_G(\theta)
&=(\cos\theta\,\mathbf I+i\sin\theta\,A)\,
B\,
(\cos\theta\,\mathbf I-i\sin\theta\,A)\nonumber\\
&=\cos^2\theta\,B
+i\cos\theta\sin\theta\,(AB-BA) \nonumber\\
   &
+\sin^2\theta\,ABA.
\end{align}
Using \(AB=-BA\) and Eq.~\eqref{eq:ABA_minusB}, this becomes
\begin{align}
U_G^\dagger(\theta)\,B\,U_G(\theta)
&=(\cos^2\theta-\sin^2\theta)\,B
+2i\cos\theta\sin\theta\,AB\nonumber\\
&=\cos(2\theta)\,\sigma_z^{1}
+2i\cos\theta\sin\theta\,
\bigl(i\,\sigma_x^{1}\sigma_y^{2}\bigr)\nonumber\\
&=\cos(2\theta)\,\sigma_z^{1}
-\sin(2\theta)\,\sigma_x^{1}\sigma_y^{2}.
\end{align}
Thus,
\begin{equation}
U_G^\dagger(\theta)\,\sigma_z^{1}\,U_G(\theta)
=
\cos(2\theta)\,\sigma_z^{1}
-\sin(2\theta)\,\sigma_x^{1}\sigma_y^{2}.
\end{equation}

\paragraph{Transformation of \(\sigma_z^{2}\).}
Similarly, taking \(B=\sigma_z^{2}\), we find
\begin{align}
AB
&=\sigma_y^{1}\sigma_y^{2}\sigma_z^{2}
=\sigma_y^{1}(\sigma_y^{2}\sigma_z^{2})
=i\,\sigma_y^{1}\sigma_x^{2},\\
BA
&=\sigma_z^{2}\sigma_y^{1}\sigma_y^{2}
=-i\,\sigma_y^{1}\sigma_x^{2},
\end{align}
so again \(\{A,B\}=0\) and \(ABA=-B\). Repeating the same calculation gives
\begin{equation}
U_G^\dagger(\theta)\,\sigma_z^{2}\,U_G(\theta)
=
\cos(2\theta)\,\sigma_z^{2}
-\sin(2\theta)\,\sigma_y^{1}\sigma_x^{2}.
\end{equation}

\paragraph{Transformation of \(\sigma_z^{1}\sigma_z^{2}\).}
Using associativity and \(U_G(\theta)U_G^\dagger(\theta)=\mathbf I\), we have
\begin{align}
\bigl(U_G^\dagger(\theta)\sigma_z^{1}U_G(\theta)\bigr)
\bigl(U_G^\dagger(\theta)\sigma_z^{2}U_G(\theta)\bigr) \nonumber\\
=
U_G^\dagger(\theta)\,
\sigma_z^{1}\,
\bigl(U_G(\theta)U_G^\dagger(\theta)\bigr)\,
\sigma_z^{2}\,
U_G(\theta)\nonumber\\
=
U_G^\dagger(\theta)\,\sigma_z^{1}\sigma_z^{2}\,U_G(\theta).
\label{eq:multiplicative_conjugation}
\end{align}

It remains to show that \(\sigma_z^{1}\sigma_z^{2}\) commutes with \(A\).
Indeed,
\begin{align}
A\,\sigma_z^{1}\sigma_z^{2}
&=
\sigma_y^{1}\sigma_y^{2}\sigma_z^{1}\sigma_z^{2}
=
(\sigma_y^{1}\sigma_z^{1})(\sigma_y^{2}\sigma_z^{2})
=
(i\sigma_x^{1})(i\sigma_x^{2}) \nonumber\\
   &
=
-\sigma_x^{1}\sigma_x^{2},\\
\sigma_z^{1}\sigma_z^{2}\,A 
&=
\sigma_z^{1}\sigma_z^{2}\sigma_y^{1}\sigma_y^{2}
=
(\sigma_z^{1}\sigma_y^{1})(\sigma_z^{2}\sigma_y^{2})
=
(-i\sigma_x^{1})(-i\sigma_x^{2})\nonumber\\
   &
=
-\sigma_x^{1}\sigma_x^{2}.
\end{align}
Hence
\begin{equation}
[A,\sigma_z^{1}\sigma_z^{2}]=0.
\end{equation}
Therefore conjugation leaves \(\sigma_z^{1}\sigma_z^{2}\) invariant:
\begin{equation}
U_G^\dagger(\theta)\,\sigma_z^{1}\sigma_z^{2}\,U_G(\theta)
=
\sigma_z^{1}\sigma_z^{2}.
\end{equation}

Combining this with Eq.~\eqref{eq:multiplicative_conjugation}, we conclude that
\begin{equation}
\sigma_z^{1}(\theta)\,\sigma_z^{2}(\theta)
=
U_G^\dagger(\theta)\,\sigma_z^{1}\sigma_z^{2}\,U_G(\theta)
=
\sigma_z^{1}\sigma_z^{2}.
\end{equation}
proving the stated identities.
\end{proof}

\section{Proof of Theorem 1 by Induction on \(N\)}\label{thm:optimal_feedbak}
From the base case \(N\!=\!1\), see Section ~\ref{sec:examples}, we have $A_1\!=\!\frac{\epsilon}{2}\, a$, $B_1\!=\!\frac{\epsilon}{2}\,b$, and $\theta_1\!=\! \arctan \left(-b/a\right)$. 
Then, for the case $N\!=\!2$, the two-coupled quantum system, the Hamiltonian is
\begin{equation}
H_S =  \frac{1}{2}\mathbb{I}_{\otimes 2} + \frac{\epsilon_1}{2} \sigma_z^{1} + \frac{\epsilon_2}{2} \sigma_z^{2}  + \Delta_{z}^{12}\, \sigma_z^{1} \sigma_z^{2}, 
\end{equation}
where $\epsilon_i$ is the energy of each two-level system $i$, and $\Delta_z^{12}$ represents the coupling strength between the systems. In what follows, we take $\Delta_z^{12}\!\equiv\!\Delta$. 
Applying the unitary feedback of the form $U_j(\theta_j) = e^{-i \theta_j \sigma_y^{j} / 2}, (j\!=\!1,2)$; the system state after the feedback reads,
\[
\rho_F = U_1(\theta_1) U_2(\theta_2) \rho_{M_k} U_2^\dagger(\theta_2) U_1^\dagger(\theta_1).
\]
The corresponding energy of the system is $E_F\!=\!\operatorname{Tr} [\rho_F H_S ]\!=\!\operatorname{Tr} [ \rho_M H_F ]$. 
 The feedback Hamiltonian can be written as
\begin{equation}
\begin{aligned}
H_F &= U_1^\dagger(\theta_1)\,U_2^\dagger(\theta_2)\,H_S\,U_2(\theta_2)\,U_1(\theta_1)\\
&= \frac{\epsilon_1}{2}\bigl(\sigma_z^{1}\cos\theta_1-\sigma_x^{1}\sin\theta_1\bigr)
+\frac{\epsilon_2}{2}\bigl(\sigma_z^{2}\cos\theta_2-\sigma_x^{2}\sin\theta_2\bigr)\\
&\quad+ \Delta\bigl(\sigma_z^{1}\cos\theta_1 -\sigma_x^{1}\sin\theta_1\bigr)
\bigl(\sigma_z^{2}\cos\theta_2-\sigma_x^{2}\sin\theta_2\bigr).
\end{aligned}
\end{equation}

Expanding the interaction term:
\begin{equation}
\begin{aligned}
&\left( \sigma_z^{1} \cos \theta_1 - \sigma_x^{1} \sin \theta_1 \right)
\left( \sigma_z^{2} \cos \theta_2 - \sigma_x^{2} \sin \theta_2 \right) \\
&= \sigma_z^{1} \sigma_z^{2} \cos \theta_1 \cos \theta_2
- \sigma_z^{1} \sigma_x^{2} \cos \theta_1 \sin \theta_2 \\
&\quad
- \sigma_x^{1} \sigma_z^{2} \sin \theta_1 \cos \theta_2
+ \sigma_x^{1} \sigma_x^{2} \sin \theta_1 \sin \theta_2.
\end{aligned}
\end{equation}

Define the following traces:
\[
\langle \sigma_z^{(1)}\rangle = a_1, 
\langle \sigma_x^{(1)}\rangle = b_1,
\langle \sigma_z^{(2)}\rangle = a_2,
\langle \sigma_x^{(2)}\rangle = b_2,
\]
and
\[
\langle \sigma_z^{1}\sigma_z^{2}\rangle\!=\!c_{zz}, 
\langle \sigma_z^{1}\sigma_x^{2}\rangle\!=\!c_{zx}, 
\langle \sigma_x^{1}\sigma_z^{2}\rangle\!=\!c_{xz}, 
\langle \sigma_x^{(1)}\sigma_x^{(2)}\rangle = c_{xx}.
\]

Therefore, the energy after feedback is
\begin{equation}
\begin{aligned}
E_F = & \operatorname{Tr} [ \rho_M H_F ] \\
=& \frac{\epsilon_1}{2} (a_1 \cos \theta_1 - b_1 \sin \theta_1)
+ \frac{\epsilon_2}{2} (a_2 \cos \theta_2 - b_2 \sin \theta_2) \\
&+ \Delta \cos \theta_1\left[ c_{zz}  \cos \theta_2 - c_{zx}  \sin \theta_2\right] \\
&-\Delta \sin \theta_1 \left[c_{xz}  \cos \theta_2 - c_{xx}  \sin \theta_2 \right].
\end{aligned}
\label{eq:2QME_Ef}
\end{equation}


To achieve our goal of optimal performance of the feedback measurement-based engine, we determine the stationary points of \( E_F \) with respect to \( \theta_1 \) and \( \theta_2 \), and then select the branch that minimizes \(E_F\). Thus, we obtain;
\begin{equation}
\begin{aligned}
\frac{\partial E_F}{\partial \theta_1}
=& -\frac{\epsilon_1}{2} \left(a_1 \sin \theta_1 + b_1 \cos \theta_1\right)\\
&-\Delta \sin \theta_1 \left[c_{zz}  \cos \theta_2 - c_{zx}  \sin \theta_2\right]\\
&-\Delta \cos \theta_1 \left[c_{xz}  \cos \theta_2 - c_{xx}  \sin \theta_2 \right].
\end{aligned}
\label{dEtheta1}
\end{equation}

Similarly, for \(\theta_2 \):
\begin{equation}
\begin{aligned}
\frac{\partial E_F}{\partial \theta_2}
=& -\frac{\epsilon_2}{2} \left(a_2 \sin \theta_2 + b_2 \cos \theta_2\right)\\
&-\Delta \sin \theta_2 \left[c_{zz} \cos \theta_1 - c_{xz} \sin \theta_1 \right]\\
&-\Delta \cos \theta_2 \left[c_{zx} \cos \theta_1 - c_{xx} \sin \theta_1 \right].
\end{aligned}
\label{dEtheta2}
\end{equation}


Solving Eqs.~(\ref{dEtheta1}) and (\ref{dEtheta2}) simultaneously, i.e.
\[
\frac{\partial E_F}{\partial \theta_1}=0,
\qquad
\frac{\partial E_F}{\partial \theta_2}=0,
\]
gives the coupled stationary-point equations. From Eq.~(\ref{dEtheta1}), we obtain
\begin{equation}
A_1\sin\theta_1 + B_1\cos\theta_1 = 0,
\end{equation}
with
\begin{align}
A_1 &= \frac{\epsilon_1}{2} a_1 + \Delta \left(c_{zz} \cos \theta_2 - c_{zx} \sin \theta_2 \right), \label{appTheta1}\\
B_1 &= \frac{\epsilon_1}{2} b_1 + \Delta \left(c_{xz} \cos \theta_2 - c_{xx} \sin \theta_2 \right).
\label{appTheta1b}
\end{align}
Similarly, from Eq.~(\ref{dEtheta2}) one gets
\begin{equation}
A_2\sin\theta_2 + B_2\cos\theta_2 = 0,
\end{equation}

with: 
\begin{align}
A_2 &= \frac{\epsilon_2}{2} a_2 + \Delta \left(c_{zz} \cos \theta_1 - c_{xz} \sin \theta_1 \right),\\
B_2 &= \frac{\epsilon_2}{2} b_2 + \Delta \left(c_{zx} \cos \theta_1 - c_{xx} \sin \theta_1 \right).
\label{appTheta2}
\end{align}

Solving for the angles, \( \theta_1 \) and \( \theta_2 \), from equations \ref{appTheta1} - \ref{appTheta2}, we obtain
$\tan \theta_1\!=\!\frac{-B_1}{A_1}, \quad \tan \theta_2\!=\!\frac{-B_2}{A_2}$.
Thus, the stationary feedback angles are:

\begin{equation}
\theta_1^\ast = \operatorname{atan2}(-B_1, A_1)
\qquad (\mathrm{mod}\,\pi),
\end{equation}
and similarly for $\theta_2.$

We note that \(A_1\) and \(B_1\) depend on \(\theta_2\),  while \(A_2\) and \( B_2 \) depend on \( \theta_1 \), which makes these equations coupled. 


\emph{Inductive Step (\(N\to N+1\))} -- Assume that the claim (including explicit forms of \(A_j\) and \(B_j\)) is true for any \(N\)-coupled qubits (e.g, shifted Ising-like Hamiltonian) as a working substance of a quantum measurement-based engine. Now considering an Ising-type Hamiltonian of the form
\begin{equation}
\begin{aligned}
H_Q^{(N+1)} 
\;=\;
\sum_{j=1}^{N+1} \frac{\epsilon_j}{2}\,\sigma_z^{j}
+\;
\sum_{j<k}^{N+1} \Delta_z^{jk}\,\sigma_z^{j}\,\sigma_z^{k},
\end{aligned}
\label{eq:HQ_Nplus1}
\end{equation}
where we have neglected the term $\frac{1}{2}\,I_{\otimes(N+1)}$.
Let \(\rho_M^{(N+1)}\) be the post-measurement state and apply a local feedback unitary of the form
\begin{equation}
U(\theta_1,\dots,\theta_{N+1})
\;=\;\bigotimes_{j=1}^{N+1} \exp\!\Bigl(-\,\frac{i\,\theta_j}{2}\,\sigma_y^{j}\Bigr).
\label{eq:local_feedback_unitaries_Nplus1}
\end{equation}
After applying the feedback, we obtain the energy as
\begin{equation}
\begin{aligned}
E_F^{(N+1)}(\theta_1,\dots,\theta_{N+1})
\!=\!&
\mathrm{Tr}\Bigl[\rho_M^{(N+1)}\,U^\dagger\,H_Q^{(N+1)}\,U\Bigr],
\end{aligned}
\label{eq:post_feedback_Nplus1}
\end{equation}
where we have use $\rho_F^{(N+1)}\!=\!
U\,\rho_M^{(N+1)}\,U^\dagger$.
Since the unitary,
$\exp\!(-i\,\theta_j\,\sigma_y^{j}/2)$, 
acts only on system \(j\), the derivative
\begin{equation}
\frac{\partial}{\partial\theta_j}
\Bigl(U^\dagger\,{H}_Q^{(N+1)}\,U\Bigr),
\label{eq:partial_derivative_HQ}
\end{equation}
affects terms where coefficients of \(\sigma_z^{j}\) appear: the Zeeman piece \(\frac{\epsilon_j}{2}\sigma_z^{j}\), plus any Ising coupling terms \(\Delta_z^{jk}\sigma_z^{j}\sigma_z^{k}\). Each derivative yields a linear combination of \(\sin\theta_j\) and \(\cos\theta_j\), with coefficients (expectation values in \(\rho_M^{(N+1)}\)) that come from
the maps \(\sigma_{z}^{j}\mapsto \sigma_{z}^{j}\cos\theta_j - \sigma_{x}^{j}\sin\theta_j\). 
Hence, the stationary condition
\[
\frac{\partial E_F^{(N+1)}}{\partial \theta_j}=0
\]
implies
\begin{equation}
A_j(\{\theta_{k\neq j}\})\,\sin\theta_j
+
B_j(\{\theta_{k\neq j}\})\,\cos\theta_j
= 0,
\label{eq:stationarity_linear}
\end{equation}
where
\[
A_j(\{\theta_{k\neq j}\})
=
\frac{\epsilon_j}{2}\,\Bigl\langle \widetilde{\sigma_{z}^{j}}\Bigr\rangle
+\sum_{k\neq j}\Delta_z^{jk}\,\Bigl\langle\widetilde{\sigma_{z}^{j}\sigma_{z}^{k}}\Bigr\rangle,
\label{eq:A_j_induction}
\]
and
\[
B_j(\{\theta_{k\neq j}\})
=
\frac{\epsilon_j}{2}\,\Bigl\langle \widetilde{\sigma_{x}^{j}}\Bigr\rangle
+\sum_{k\neq j}\Delta_z^{jk}\,\Bigl\langle\widetilde{\sigma_{x}^{j}\sigma_{z}^{k}}\Bigr\rangle.
\label{eq:B_j_induction}
\]

Here, the “\(\widetilde{\phantom{\sigma}}\)” indicates partial rotations on the other qubits. Each \(\langle\widetilde{\sigma_{\alpha}^{j}\sigma_{\beta}^{k}}\rangle\) means: take the post-measurement state \(\rho_M\), apply the partial rotations associated with all qubits except the one(s) whose operator(s) are being labeled, and then evaluate the corresponding expectation value. The same rule applies to single-qubit operators such as \(\widetilde{\sigma_z^j}\).

Solving Eq.~(\ref{eq:stationarity_linear}) yields
\begin{equation}
\tan\theta_j
=
\frac{-B_j(\theta_1,\dots,\theta_{j-1},\theta_{j+1},\dots,\theta_{N+1})}
{A_j(\theta_1,\dots,\theta_{j-1},\theta_{j+1},\dots,\theta_{N+1})},
\label{eq:tangent_equation}
\end{equation}
or equivalently
\begin{equation}
\theta_j^\ast
=
\operatorname{atan2}\!\bigl(-B_j,A_j\bigr)
\qquad (\mathrm{mod}\,\pi).
\end{equation}
By the inductive hypothesis, appending the \((N+1)\)-th system preserves the same structure of \(A_j\) and \(B_j\). Therefore, the claimed form holds for all \(N\). \(\quad\Box\)

\section{Global (non-local) feedback}\label{global_appendix}
In this section, we evaluate the energy after feedback using a global (non-local) feedback protocol. Then, we compare the resulting maximum work extraction with the case of local feedback when implemented in a measurement-based engine with a coupled two-level system as a working substance.

Using the identities from Section~\ref{sec:global}, we obtain
\begin{align}
E_F^{(G)}(\theta)
&=
\frac{\epsilon_1}{2}\,
\Tr\!\left[\rho_M\bigl(\cos(2\theta)\sigma_z^1-\sin(2\theta)\sigma_x^1\sigma_y^2\bigr)\right]
\nonumber\\
&\quad
+
\frac{\epsilon_2}{2}\,
\Tr\!\left[\rho_M\bigl(\cos(2\theta)\sigma_z^2-\sin(2\theta)\sigma_y^1\sigma_x^2\bigr)\right] \nonumber\\
&
+
\Delta\,\Tr(\rho_M \sigma_z^1\sigma_z^2)
\nonumber\\
&=
\frac{\epsilon_1}{2}\bigl(\cos(2\theta)\,a_1-\sin(2\theta)\,c_{xy}\bigr) \nonumber\\
&
+
\frac{\epsilon_2}{2}\bigl(\cos(2\theta)\,a_2-\sin(2\theta)\,c_{yx}\bigr)
+\Delta\,c_{zz},
\label{eq:EF_global_correct}
\end{align}
where \(c_{xy}=\langle \sigma_x^1\sigma_y^2\rangle\) and
\(c_{yx}=\langle \sigma_y^1\sigma_x^2\rangle\).

On the other hand, following Eq.~(\ref{eq:2QME_Ef}), the energy after local feedback is
\begin{equation}
\begin{aligned}
E_F^{(L)}
=& \frac{\epsilon_1}{2} (a_1 \cos \theta_1 - b_1 \sin \theta_1)
+ \frac{\epsilon_2}{2} (a_2 \cos \theta_2 - b_2 \sin \theta_2) \\
&+ \Delta \cos \theta_1\left[ c_{zz}  \cos \theta_2 - c_{zx}  \sin \theta_2\right] \\
&-\Delta \sin \theta_1 \left[c_{xz}  \cos \theta_2 - c_{xx}  \sin \theta_2 \right].
\end{aligned}
\end{equation}

For the restricted non-local ansatz in Eq.~\eqref{global_unitary}, the best local protocol yields a lower post-feedback energy, equivalently a larger extracted work, over the parameter range shown.


\bibliography{paper}
\end{document}